\documentclass[dvipdfm]{llncs}
\usepackage{epsfig}

\usepackage{makeidx}  % allows for indexgeneration

\begin{document}

\frontmatter          % for the preliminaries

\pagestyle{headings}  % switches on printing of running heads
\addtocmark{The $L(2,1)$-Labeling Problem on Oriented Regular Grids} % additional mark in the TOC

\newtheorem{observation}[theorem]{Observation}

\spnewtheorem*{sproof}{Sketch of Proof}{\itshape}{\rmfamily}

\newcommand{\DENSE}{%
\setlength{\labelwidth}{12pt}%
\setlength{\labelsep}{2pt}%
\setlength{\leftmargin}{\labelwidth}%
\addtolength{\leftmargin}{\labelsep}%
\setlength{\parsep}{0pt}%
\setlength{\itemsep}{0pt}%
\setlength{\topsep}{0pt}
}
\newenvironment{ditemize}{\begin{list}%
{$-$~}{\DENSE}}{\end{list}}
\font\testo = cmss10 at 12pt
\font\scr = cmss10
\font\scrscr = cmss10 at 6pt
\def\erre{\mathchoice{\hbox{\testo
I\kern-.17emR}}{
   \hbox{\testo I\kern-.17emR}}{\hbox{\scr
I\kern-.17emR}}{
   \hbox{\scrscr I\kern-.17emR}}}
\def\enne{\mathchoice{\hbox{\testo
I\kern-.17emN}}{
   \hbox{\testo I\kern-.17emN}}{\hbox{\scr
I\kern-.17emR}}{
   \hbox{\scrscr I\kern-.17emN}}}
   
\mainmatter              % start of the contributions
\title{The $L(2,1)$-Labeling Problem on Oriented Regular Grids}
\titlerunning{The $L(2,1)$-Labeling Problem on Oriented Regular Grids}  % abbreviated title (for running head)
\author{Tiziana Calamoneri\inst{1}}
\authorrunning{Tiziana Calamoneri}   % abbreviated author list (for running head)
 % additional mark in the TOC
\institute{Department of Computer Science \\
   University of Rome ``Sapienza'' - Italy\\
   via Salaria 113, 00198 Roma, Italy.\\
\email{e-mail: calamo@di.uniroma1.it}
}

\tocauthor{Tiziana Calamoneri (University of Rome `'Sapienza'')}
\institute{Department of Computer Science\\ 
University of Rome ``Sapienza'' - 
Italy\\
via Salaria 113, 00198 Roma, Italy. \\
\email{e-mail: calamo@di.uniroma1.it}}

\maketitle              % typeset the title of the contribution

%\cleardoublepage

\begin{abstract}
The $L(2,1)$-labeling of a digraph $G$ is a function $f$ from the node set  of $G$  to the set of all nonnegative integers such that $|f(x)-f(y)| \geq 2$ if $x$ and $y$ are at distance 1, and $f(x) \neq f(y)$ if $x$ and $y$ are at distance 2, where the distance from vertex $x$ to vertex $y$ is the length of a shortest dipath from $x$ to $y$. The minimum of the maximum used label over all  $L(2,1)$-labelings of $G$  is called $\vec{\lambda} (G)$.

In this paper we study the $L(2,1)$-labeling problem on squared, triangular and hexagonal grids and for them we compute the exact values of $\vec{\lambda}$.
\end{abstract}

{\bf keywords}
$L(h,k)$-labeling,
graph coloring,
digraphs,
regular grids.

%sezione:introduzione
\section{Introduction}

The {\em $L(2,1)$-labeling problem} has been proposed by Griggs and Yeh \cite{GY92} as a variation of the frequency assignment problem introduced by Hale \cite{H80}.
It refers to the frequency assignment problem of wireless networks where fixed antennas can both transmit and receive messages.
More precisely, neighbor transmitters must use frequencies at least $2$
apart, and receivers that are neighbors of the same transmitter (two hops
far) must use frequencies at least 1 apart.
This problem can be modeled as a coloring problem, where the aim is to minimize $\lambda$, i.e. the maximum used color.
Since its definition, a huge number of works have been produced on this topic (some surveys looking at the problem from different point of views are \cite{Aal01,C06,Y06}).
A natural extension, recently introduced \cite{CL03}, is the $L(2,1)$-labeling on digraphs.
The {\em $L(2,1)$-labeling of a digraph} $G$ is a function $f$ from the node set  of $G$  to the set of all nonnegative integers such that $|f(x)-f(y)| \geq 2$ if $x$ and $y$ are at distance 1, and $f(x) \neq f(y)$ if $x$ and $y$ are at distance 2, where the distance from vertex $x$ to vertex $y$ is the length of a shortest dipath from $x$ to $y$. The minimum of the maximum used label over all  $L(2,1)$-labelings of $G$  is called $\vec{\lambda} (G)$.

We underline that in the following the maximum used color will be called $\vec{\lambda}$ in the directed case and $\lambda$ in the undirected case.

To the best of the author's knowledge, only two papers deal with the $L(2,1)$-labeling of digraphs.
The first one, by Chang and Liaw \cite{CL03}, introduces the definition and  studies the problem on ditrees, surprisingly proving that $\vec{\lambda} (T) \leq 4$ for any ditree $T$, while in the non-oriented case $\lambda$ is either $\Delta +1$ or $\Delta +2$, where $\Delta$ is the maximum degree of the tree.
The second paper comes some years later \cite{Cal07} and deals with the more general $L(h,k)$-labeling problem on bipartite digraphs with conditions on the length of the longest dipath.

Both papers suggest that the most important parameter of the digraph giving information on $\vec{\lambda}$ is the length of the longest dipath, differently with respect to the non oriented case, in which the value of $\lambda$ is usually function of the maximum degree $\Delta$.

In this paper we approach the oriented $L(2,1)$-labeling problem on regular grids, i.e. squared, triangular and hexagonal grids, stressing again the importance of the length of the longest dipath. 
More precisely, we evaluate the smallest value of the length of the longest dipath such that a certain value of $\vec{\lambda}$ holds.
Nevertheless, at the end, we propose a different parameter of the graph to be kept in consideration when studying its $L(2,1)$-labeling. 
Such a parameter is the girth of the digraph, and we suspect that it makes difference if it is larger than or equal to 5.
We highlight that the previous papers could not deal with it, as they restrict the problem to trees (no girth) and to bipartite graphs whose longest dipath has length at most three (girth $\leq 4)$.

\medskip

The rest of this paper is organized as follows. 
Next section is devoted to recall definitions and useful results, to prove some simple bounds and to survey all the results known on the oriented $L(2,1)$-labeling problem. 
Sections 3, 4 and 5 focus on squared, triangular and hexagonal grids, respectively.
Finally, in Section 6 we address some conclusions, list some open problems and state an interesting conjecture. 

%%%%%%%%%%%%%%%%%%%%%%%%%%%%%%%%%%%%%%%%%%%%%%%%%%%%%%
\section{Preliminaries  and Discussion of the Results}

An {\em (oriented) $L(2,1)$-labeling} of a
digraph $D=(V, E)$ is a function $f: V \rightarrow \enne$ such that
\begin{ditemize}
\item  $|f(u)-f(v)| \geq 2$ if $(u,v) \in E$ and
\item  $f(u) \neq f(v)$ if there exists $w \in V$ such that $(u, 
w) \in E$ and $(w, v) \in E$.
\end{ditemize}

The {\em span} of an $L(2,1)$-labeling is the difference between the
largest and the smallest value of $f$, so it is not restrictive to
assume 0 as the smallest value.
We denote by
$\vec{\lambda}(D)$ (or simply $\vec{\lambda}$ when the digraph is clear from the context) the smallest integer $\sigma$
such that digraph $D$ has an $L(2,1)$-labeling of span $\sigma$.

\medskip

In this paper we study the $L(2,1)$-labeling problem on the orientations of the three regular tilings of the plane, i.e. the infinite graphs whose a portion is depicted in Figure \ref{fig.tilings}.

\begin{figure}[ht]
\begin{center}
    \epsfxsize=11cm
	\epsfbox{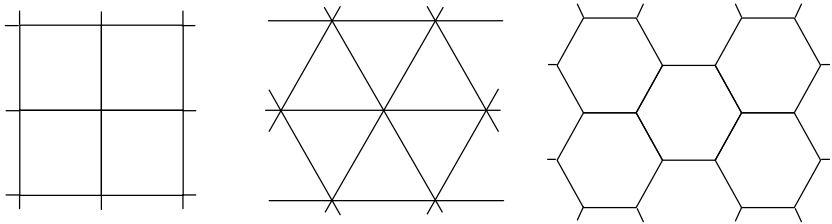}
   \caption{The three regular tilings of the plane.}
   \protect\label{fig.tilings}
\end{center}
\end{figure}

To our aims, we will show that we do not need a whole infinite tiling, but it is enough a finite portion of it. So, in the following we will call {\em grid} the subgraph of an infinite tiling induced by a finite subset of its nodes. 
As we are interested in a coloring of such a grid, in order to determine the maximum needed color, it is not restrictive to assume that the grid is a connected graph. We will call a grid {\em squared, triangular} or {\em hexagonal} if it is an induced subgraph of a squared, a triangular or a hexagonal tiling, respectively. We will use the symbols $G_{\Delta}$ and $T_{\Delta}$ to mean a grid or an infinite tiling, respectively; the values of $\Delta=3,4,6$ indicate the degree and hence the hexagonal, squared or triangular shape of the faces.

\medskip

An {\em oriented graph} is a digraph without opposite arcs. We consider {\em oriented grids}, i.e. grids whose each edge has a fixed orientation. The (undirected) graph obtained by an oriented grid by removing the orientation of all its arcs is called the {\em underlying grid}.
Observe that several oriented grids can correspond to the same underlying grid. 

All over this paper, where no differently specified, we will consider {\em non-trivial grids}, i.e. grids whose underlying grid contains at least one cycle.

The following property is an obvious consequence of the definition of $\vec{\lambda}(D)$.

\begin{theorem} \cite{Cal07}
\label{th.subgraph}
If $H$ is a subgraph of a digraph $D$, then $\vec{\lambda}(H) \leq \vec{\lambda}(D)$.
\end{theorem}

This result allows us to say that, if we are able to find a grid $G$ for which we can prove $\vec{\lambda}(G)=k$ for a certain $k$, then there exist infinite grids $G'$ such that $\vec{\lambda}(G') \geq k$, being $G'$ is any grid containing $G$ as subgrid.

\medskip

Now, we easily prove the following general results that we will exploit later.

\begin{theorem}
\label{th.upperbound}
If $D$ is a digraph and $U$ is its underlying graph, then the minimum span necessary to $L(2,1)$-label  digraph $D$ never exceeds the minimum span necessary to $L(2,1)$-label its underlying (undirected) graph, i.e. $\vec{\lambda}(D) \leq \lambda(U)$.
\end{theorem}

\begin{proof}
Two nodes at distance 1 in $U$ are at distance 1 in $D$, too, while two nodes at distance 2 in $U$ are at (oriented) distance $\geq 2$ (possibly $\infty$) in $D$ because of the orientation of the arcs. So, the statement immediately follows. \qed
\end{proof}

We remind that, for the (undirected) regular tilings of the plane, it holds that  $\lambda(T_{\Delta})=\Delta+2$, $\Delta = 3,4,6$ \cite{CP04} so, applying Theorems \ref{th.subgraph} and \ref{th.upperbound}, we get that, for any (oriented) grid $G_{\Delta}$,  $\vec{\lambda}(G_{\Delta}) \leq \Delta +2$, $\Delta=3,4,6$. These values provide us upper bounds on $\vec{\lambda}$ for the grids of all three types.

\begin{theorem}
\label{th.lowerbound}
If $D$ is a digraph and the length of the maximum dipath is either 2 or 3 then $\vec{\lambda}(D) \geq 3$; if the length of the maximum dipath is $ \geq 4$ then $\vec{\lambda}(D) \geq 4$.
\end{theorem}

\begin{proof}
The claim follows from Theorem \ref{th.subgraph} and from the $L(2,1)$-labeling of a path \cite{GY92,Y90}, that can trivially be transferred to label a dipath of the same length. 
We remind that path $P_3$ of length 2 can be optimally labeled as $1 \rightarrow 3 \rightarrow 0$, path $P_4$ of length 3 as $1 \rightarrow 3 \rightarrow 0 \rightarrow 2$ and a path of length equal to or larger than 4 can be labeled as $0 \rightarrow 2 \rightarrow 4 \rightarrow 0 \rightarrow 2 \rightarrow \ldots $. \qed
\end{proof}

We say a digraph to be {\em bipartite} if its underlying graph is bipartite.
It is easy to see that both a squared and a hexagonal oriented grid are bipartite digraphs while a triangular oriented grid is not.

\medskip

We recall here some results on bipartite digraphs that will be useful in the following.

\begin{theorem} \cite{Cal07}
\label{th.bipartitel=2}
For any bipartite digraph $D$ whose longest dipath has length 2, $\vec{\lambda}(D)=3$.
\end{theorem}

\begin{theorem} \cite{Cal07}
\label{th.bipartitel=3}
For any bipartite digraph $D$ whose longest dipath has length 3,  $3 \leq \vec{\lambda}(D) \leq 4$.
\end{theorem}

For non-bipartite digraphs, it is known the value of $\vec{\lambda}(D)$ only when the longest dipath has length 2, according to the following theorem.

\begin{theorem} \cite{Cal07}
\label{th.generall=2}
For any non-bipartite digraph $D$ whose longest dipath has length 2, $\vec{\lambda}(D)=4$.
\end{theorem}

{\bf Discussion of the results. } In the following table we summarize all the known results concerning the $L(2,1)$-labeling of digraphs. Namely, all the bounds already known in the literature are those dealing with trees and bipartite graphs and those found for graphs having the longest dipath of length either 1 or 2; all the other results can be found in the present paper. 
The bounds with a question mark are not proved and we leave an open problem about them.
Symbol $l$ indicates the length of the longest dipath in the digraph.

\medskip

\hspace{-.4cm}
\begin{tabular}{| l | c | c | c | c | c | c |}
\hline
& general & trees & bipartite & squared & triangular & hexagonal \\
&  digraphs & & digraphs & grids & grids & grids \\
& & & & $\vec{\lambda} \leq 6$ & $\vec{\lambda} \leq 8$ & $\vec{\lambda} \leq 5$ \\
\hline
$l=1$ & $\vec{\lambda}=2$ & $\vec{\lambda}=2$ & $\vec{\lambda}=2$ &
     $\vec{\lambda}=2$  & $\vec{\lambda}=2$ & $\vec{\lambda}=2$  \\
\hline
$l=2$ & $\vec{\lambda}=4$ &   $\vec{\lambda}=3$ & $\vec{\lambda}=3$ &
    $\vec{\lambda}=3$ & $\vec{\lambda}=4$ & $\vec{\lambda}=3$  \\
\hline
$l=3$ $\vec{\lambda} \geq 3$ & & $\vec{\lambda} =3, 4$ &  $\vec{\lambda} =3, 4$  & 
    $\vec{\lambda}=3,4$ &  $\vec{\lambda}=3,4,5;\geq 6?$ &  $\vec{\lambda}=3,4$ \\
\hline    
$l \geq 4$ $\vec{\lambda} \geq 4$ & &  $\vec{\lambda}=4$ & & $\vec{\lambda}=4,5;6?$ &  $\vec{\lambda}=4,5,6; \geq 7?$ &  $\vec{\lambda} = 4; 5?$ \\
\hline
$l \geq 5$ $\vec{\lambda} \geq 4$ & &  '' & & '' &  $\vec{\lambda}=4,5,6,7; 8?$ &     ''   \\
\hline
$l \geq 6$ $\vec{\lambda} \geq 4$ & &  '' & & '' &  $\vec{\lambda}=4,5,6,7,8$ &  '' \\
\hline
$l \geq 8$ $\vec{\lambda} \geq 4$ & &  '' & & $\vec{\lambda}=4,5,6$ &  '' &  '' \\
\hline

\end{tabular}

\medskip

Some of the results in this paper are obtained {\em by construction}, i.e. we prove that any grid $G$ with a certain length of its longest dipath and a certain value of $\vec{\lambda}$ must have some fixed properties, that we exploit to produce a witness graph.

The remaining results are obtained {\em by example}, i.e. a special grid $G$ with given $\vec{\lambda}(G)$ is provided.

We underline that we do not aim at finding the minimal grid with a certain value of $\vec{\lambda}$, as it is not our interest. 
%Nevertheless, we have tried to provide completely different grids for increasing values of the length of the longest dipath, for any fixed value of $\vec{\lambda}$, instead of trivially appending dipaths to the tail of the existing longest dipath to increase the length and keep the same $\vec{\lambda}$.
%****questo non e' vero sempre. verificare*****

From the above table we can conclude that in the $L(2,1)$-labeling problem of oriented squared and triangular grids we leave a small open problem consisting in estimating the minimum length of the longest path in order to have a certain value of $\vec{\lambda}$; finally, for what concerns the hexagonal grid, we would have expected to find a grid with $\vec{\lambda}=5$ but we did not, nay we suspect that $\vec{\lambda}(T_3) \leq 4$ for any orientation of $T_3$'s edges. In the last section we will explain the reasons of this convinction.

%%%%%%%%%%%
\section{Squared Grids}

First remind that any squared grid $G_4$ is a bipartite graph, so if the longest dipath has length 3 then, from Theorem \ref{th.bipartitel=3}, we have $3 \leq \vec{\lambda}(G_4) \leq 4$.

The following theorem shows that both values are attainable.
\begin{theorem}
There exist two oriented squared grids $G_4$ and $G_4'$ whose longest dipath is 3 such that $\vec{\lambda}(G_4)=3$ and $\vec{\lambda}(G'_4)=4$.
\end{theorem}

\begin{proof}
Grid $G_4$ with a possible labeling requiring colors from 0 to 3 is shown in Figure \ref{fig.squared3}.a.
It is easy to see that the provided labeling is feasible.
%Observe (for the sake of brevity, we will omit this observation in the following although it can always be done) that the existence of a grid $G_4$ having $\vec{\lambda}(G_4)=3$ implies the existence of infinite such grids, in view of Theorem \ref{th.subgraph}, infinite grids.

Now, observe that if the dipath of length 3 goes around a square of the grid (see Figure \ref{fig.squared3}.b), then it is not possible anymore to use on it the unique optimal coloring $1 \rightarrow 3 \rightarrow 0 \rightarrow 2$, because the nodes labeled 1 and 2 come in adjacency and this contradicts the definition of $L(2,1)$-labeling. So, for such grids it must be $\vec{\lambda} \geq 4$. In view of Theorem \ref{th.bipartitel=3}, $\vec{\lambda} \leq 4$ and the thesis follows.
\qed
\end{proof}

\begin{figure}[ht]
\begin{center}
    \epsfxsize=7cm
	\epsfbox{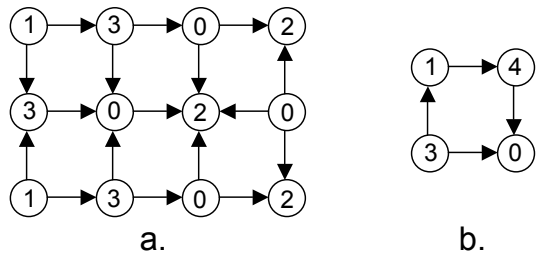}
   \caption{Two squared grids with the length of the longest dipath equal to three having: a. $\vec{\lambda}=3$ and b. $\vec{\lambda}=4$. Optimal labelings are depicted.}
   \protect\label{fig.squared3}
\end{center}
\end{figure}

If the longest dipath has length at least 4 then the graph contains a $P_4$ and hence $\vec{\lambda} \geq 4$.
From the other hand, it holds $\vec{\lambda} \leq \lambda =6$.
The following theorem shows that all three values are attainable, but the author  is convinced that no orientation of a squared grid having $\vec{\lambda}=6$ has longest dipath shorter than 8.
in other words, the suspect is that any squared grid $G_4$ such that $\vec{\lambda}(G_4)\geq 6$ must have its longest dipath of length at least 8.

\begin{theorem}
There exist three oriented squared grid $G_4$ and $G'_4$ whose longest dipath is 4 such that  $\vec{\lambda}(G_4)=4$ and  $\vec{\lambda}(G'_4)=5$.
\end{theorem}

\begin{figure}[ht]
\begin{center}
    \epsfxsize=10.5cm
	\epsfbox{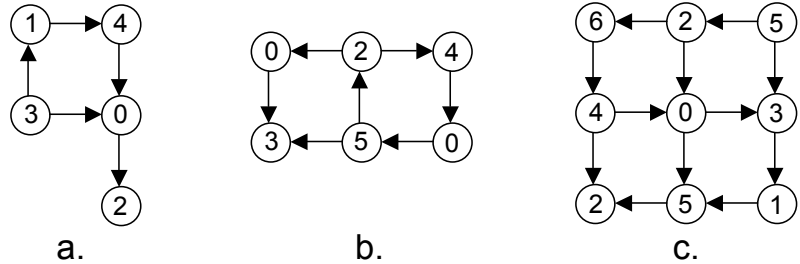}
   \caption{a. and b. Two squared grids with the length of the longest dipath equal to four having $\vec{\lambda}=4$ and $\vec{\lambda}=5$, respectively; c.a squared grid with the length of the longest dipath equal to eight having  $\vec{\lambda}=6$. Optimal labelings are depicted.}
   \protect\label{fig.squared4}
\end{center}
\end{figure}

\begin{proof} %controllata
We will prove this theorem by example. Consider the grids shown in Figure \ref{fig.squared4}.a. and \ref{fig.squared4}.b. The labelings are all feasible. 

In order to prove that they are optimal, observe that the one in Figure \ref{fig.squared4}.a is labeled with the minimum number of colors and hence is necessarily optimal.

For what concerns the grid in Figure \ref{fig.squared4}.b, the proof is leaded in an exhaustive case by case way; more precisely, it is possible to check that any labeling using colors from 0 to 4  is not feasible. We omit any further detail for the sake of brevity.
\qed
\end{proof}

\begin{theorem}
There exists an oriented squared grid $G_4$ whose longest dipath is 8 such that  $\vec{\lambda}(G_4)=6$.
\end{theorem}

\begin{proof}
Consider the grid in Figure \ref{fig.squared4}.c. For it we prove that colors from 0 to 5 are not sufficient to label it. In order to do that, let us name the nodes of the graph from top to bottom and from left to right with letters from $a$ to $i$, so the nodes in the first row are named $a$, $b$ and $c$, the nodes in the second row are named $d$, $e$ and $f$ and the nodes in the third row $g$, $h$ and $i$.
By contradiction, let us suppose that it is possible to label this grid with colors from 0 to 5.
First observe that nodes $e$, $b$, $d$, $f$ and $h$ get different colors, as they are at mutual distance 1 or 2.
If $e$ uses a color among 1, 2, 3 or 4, then no enough colors are available because it forbids the use of three colors for $b$, $d$, $f$ and $h$. So $e$ must have either the first or the last color. Without loss of generality, let 0 be the color assigned to $e$. So, nodes $b$, $d$, $f$ and $h$ will receive colors 2, 3, 4 and 5 in some order.
Let us study the color of $d$. 
\begin{itemize}
\item
If $d$ gets color 2 and $b$ is assigned color either 4 or 5, then there are no available colors for $a$. Hence $b$ is colored with 3 and the colors assigned to $a$ and $g$ must necessarily be 5 and 4, respectively. But node $h$ has not a feasible color to be assigned.
\item
If $d$ gets color 3 and $b$ is assigned color 2, then $a$ must receive color 5 and $g$ color 1 and $f$ can have colors either 4 or 5; in both cases, $c$ do not have any feasible color. If $b$ gets color 4, then $a$, $c$, $g$ and $f$ take colors 1, 2, 5 and 5, respectively. Node $i$ can receive either 1 or 3; in both cases, for node $h$ there are no feasible colors.
So $b$ is labeled with 5, implying that $a$ and $g$ take colors 1 and 5, respectively. But, in this way, no colors are available for $h$.
\item
If $d$ gets color 4 and $b$ is assigned color 2, once again there are no available colors for $a$. So, $b$ is colored with either 3 or 5. If $b$ receives color 3, then nodes $a$, $c$, $g$ and $h$ receive colors 1, 5, 2 and 5, respectively. Node $i$ gets color either 1 or 3, and in both cases there are no feasible colors for $f$. It follows that $b$ must be colored with 5 and, in this case $a$ can receive either 1 or 2. If $a$ is labeled with 1 then $g$ is labeled with 2 and there are no colors available for $h$. Hence, $a$ must be labeled 2, and consequently $g$, $h$, $f$ and $i$ take colors 1, 3, 2 and 5; but now there are no feasible colors for $c$.
\item
If, finally, $d$ gets color 5 then $b$ can assume labels 2, 3 or 4. If $b$ has label 2, then no colors are available for $a$. If $b$ has label 3 then $a$ is obliged to assume color 1 and, consequently, $c$, $f$ and $h$ colors 5, 2 and 4, respectively. It follows that no colors remain available for $i$.
So, it must be that $b$ takes color 4 and $a$ takes color either 1 or 2; both these two colors for $a$ lead again to no color for  some node. 
\end{itemize}
A contradiction raised from assuming that colors from 0 to 5 were enough.
\qed
\end{proof}

Observe that the grid provided in the previous proof has a central node whose four neighbors are all at mutual distance two, and even their common neighbors are at mutual distance two. This is a sufficient condition to have $\vec{\lambda}=6$.
By exhaustively enumerating all grids with 9 nodes having this property, it turns out that the longest dipath has always length 8.

%So, for oriented squared grids, it is completely solved the problem of evaluating the smallest value of the length of the longest dipath such that a certain value of $\vec{\lambda}$ holds.

%%%%%%%%%%%%
\section{Triangular Grids}
%%%%%%%%%%%%

%%%%%%%%%%%%%%%%%%%
\subsection{Longest dipath of length 3}

If the longest dipath of a triangular grid  has length 3 then $\vec{\lambda} \geq 3$ as the graph  contains a $P_4$. 
From the other hand, $\vec{\lambda} \leq \lambda =8$.
So, for any triangular grid $G_6$,  $\vec{\lambda}(G_6)$ must be included between these values.

%lambda=3
In particular, notice that the only way to label a $P_4$ using colors from 0 to 3 is $1 \rightarrow 3 \rightarrow 0 \rightarrow 2$ (or vice-versa $2 \rightarrow 0 \rightarrow 3 \rightarrow 1$).
In both cases, if we consider a node adjacent to both the first and the second node of the path (no matter the orientation of the edges), there are no available colors for it between 0 and 3. 
The same reasoning can be done for the second and the third nodes, and for the third and the fourth nodes.
It is neither possible that two consecutive edges of $P_4$ lie on the sides of the same triangle, otherwise  adjacent colors would be at oriented distance 1.
It follows that the only triangular grid whose longest dipath is 3 having $\vec{\lambda}=3$ is in fact $P_4$.

\medskip

%%%%%si pu˜ tagliare fino al teorema nella versione breve%%%%
%lambda=4
For what concerns $\vec{\lambda}=4$, we can do the following considerations.
Refer to the non-oriented grid in Figure \ref{fig.triangular}.a, in which nodes are named from $a$ to $g$.
In order to allow to some of its orientations to have $\vec{\lambda}=4$, only one labeling is feasible (up to symmetries) and it is the one assigning 2 to nodes $a$, $c$ and $e$, 4 to nodes $b$, $d$ and $f$ and 0 to node $g$. Any orientation for which such labeling is feasible cannot have a dipath of length 2 between $a$ and $c$, so either $a \rightarrow b \leftarrow c$ or $a \leftarrow b \rightarrow c$. As these two cases are symmetric, it is not restrictive to assume that $a \rightarrow b \leftarrow c$. 
With the same reasoning, it is not possible to have a dipath of length 2 between $a$ and $c$, $c$ and $e$, $d$ and $f$, $e$ and $a$.
Hence, we have that $c \rightarrow d \leftarrow e \rightarrow f \leftarrow a$ (see Figure \ref{fig.triangular}.b).
It is neither possible to have a dipath of length 2 between $a$ and $c$ passing through $g$, so either $a \rightarrow g \leftarrow c$ or $a \leftarrow g \rightarrow c$. In both cases, the direction of edge $(g,e)$ is fixed as it cannot exists a dipath of length 2 between $a$ and $e$ passing through $g$.
A similar reasoning holds for nodes $b$, $d$ and $f$.
By combining all the possibilities, only the four orientations in Figures \ref{fig.triangular}.c--f (up to simmetries) are possible. The first three orientations have a longest dipath of length 2 while the last one of length 4.

\begin{figure}[ht]
\begin{center}
    \epsfxsize=10cm
	\epsfbox{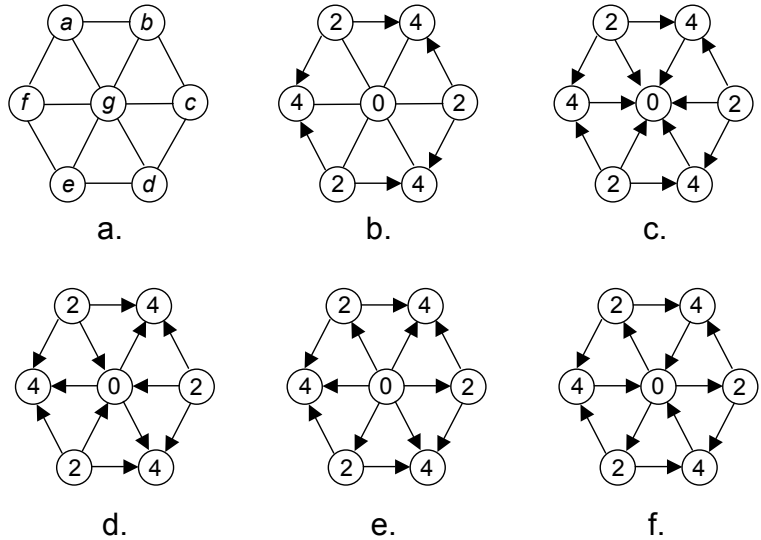}
   \caption{Reasoning for triangular grid with $\vec{\lambda}=4$. a. Underlying graph; b. unique possible labeling and compulsory orientations of some edges; c.--f. all possible orientations (up to symmetries) of graph in a. having $\vec{\lambda}=4$.}
   \protect\label{fig.triangular}
\end{center}
\end{figure}

It follows that a grid with a longest path of length 3 having $\vec{\lambda}=4$ cannot be the graph induced by a node and its six adjacent nodes.

\begin{theorem}
There exist triangular grids $G_6$ and $G_6'$, whose longest dipath is 3 such that $\vec{\lambda}(G_6)=4$ and $\vec{\lambda}(G_6')=5$.
\end{theorem}

\begin{proof}
%lambda=4
The grid $G_6$ depicted in  Figure \ref{fig.triangular34}.a. has a longest dipath of length 3.
Exploiting the previous considerations, it follows that  it requires $\vec {\lambda}(G_6) =4$.

%lambda=5
Consider now the triangular grid $G_6'$ in Figure \ref{fig.triangular34}.b.
Let $x$ be its upper node, $w$ and $y$ the middle nodes, and $z$ be the bottom node. By contradiction, let us assume the colors from 0 to 4 are enough to label such a graph.
If $x$ is colored with 0, then $y$ and $w$ must be colored with 2 and 4 in some order, and there are no feasible colors for $z$.
A similar reasoning holds if $x$ is colored with 4.
If, finally, $x$ is labeled with a color among 1, 2, 3, then there are no feasible colors for both $y$ and $w$. So, $\vec{\lambda}(G_6')=5$.
%%lambda=6
%****questo e' sbagliato: l=4. Non ho ancora trovato un grafo con l=3 e lambda=6****
%Finally, consider the triangular grid $G_6'''$  in Figure \ref{fig.triangular3}. The depicted labeling is feasible, the longest dipath has length 3 and, in a exhaustive way, it is possible to prove that $\vec{\lambda}=6$. We omit this case by case proof for the sake of brevity.
\qed
\end{proof}

\begin{figure}[ht]
\begin{center}
    \epsfxsize=11.5cm
	\epsfbox{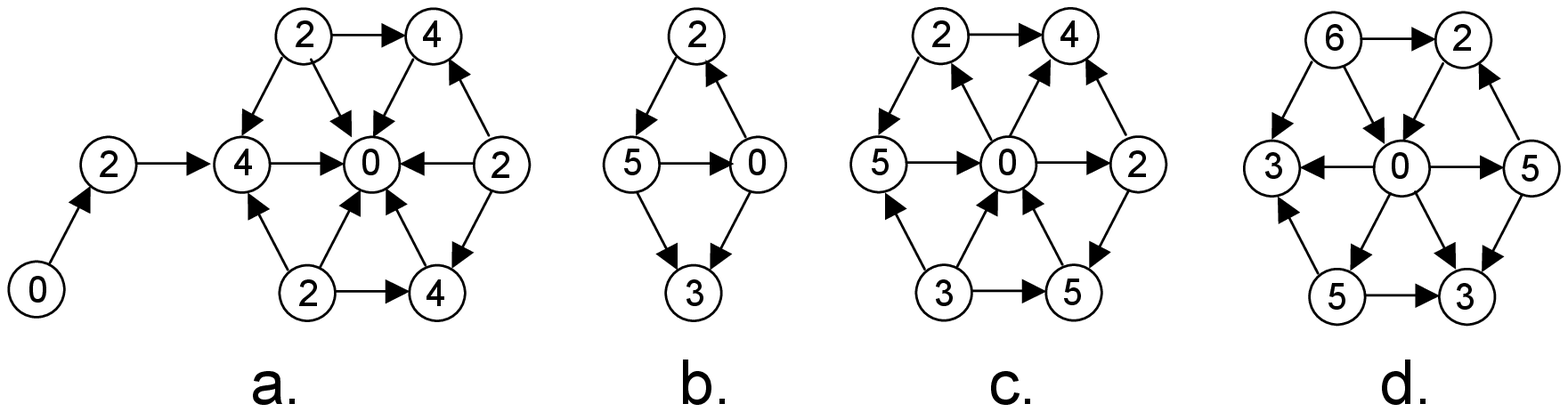}
   \caption{a. and b. Two triangular grids with the length of the longest dipath equal to three having $\vec{\lambda}=4$ and $\vec{\lambda}=5$, respectively;
   c. and d. two triangular grids with the length of the longest dipath equal to four having $\vec{\lambda}=5$ and $\vec{\lambda}=6$, respectively. Optimal labelings are depicted.}
   \protect\label{fig.triangular34}
\end{center}
\end{figure}

Observe that, in order to obtain a triangular grid requiring $\vec{\lambda}=6$, one way consists in considering a grid constituted by a node and 5 of its 6 adjacent nodes, such that each pair of these 5 nodes is connected by a path of length 2.
It turns out that any such orientation has a longest dipath of length at least equal to five. This is a sufficient condition to have $\vec{\lambda}=6$.

The author is convinced that It does not exist any triangular grid $G_6$ whose longest dipath has length 3 such that $\vec{\lambda}(G_6)\geq 6$;
in other words, the suspect is that any triangular grid $G_6$ such that $\vec{\lambda}(G_6)\geq 6$ must have its longest dipath of length at least 4.

%%%%%%%%%%%%%%%%%%%%%%%
\subsection{Longest dipath of length $\geq 4$}
\label{ssec.triangular5}

If a grid has a longest dipath of length 4, then at least colors from 0 to 4 are necessary for it. We will show that there exist triangular grids whose longest dipath has length 4 with $\vec{\lambda}$ equal to either 4, or 5 or 6.

\begin{theorem}
There exist triangular grids $G_6$, $G_6'$ and $G_6''$, whose longest dipath is 4 such that $\vec{\lambda}(G_6)=4$, $\vec{\lambda}(G_6')=5$ and $\vec{\lambda}(G_6'')=6$.
\end{theorem}

\begin{proof}
%lambda=4
$G_6$ is the grid shown in Figure \ref{fig.triangular}.f. We have already shown that the provided labeling is optimal.

%lambda=5 and =6
Consider the grids $G_6'$ and $G_6''$ shown in Figure \ref{fig.triangular34}.c and \ref{fig.triangular34}.d.
The shown $L(2,1)$-labelings are feasible and their optimality can be proved by means of a case by case reasoning, showing that no labeling using colors from 0 to 5 and to 6, respectively, is feasible. \qed
%lambda=7
%??
\end{proof}

For what concerns the values of $\vec{\lambda}$ from 4 to 6 when the longest dipath has length $\geq 5$, we can replicate the method we used in Figure \ref{fig.triangular34}.a. and append a dipath to a node of the grids having a short longest dipath in order to get grids having longest dipath of length 5 keeping the same value of $\vec{\lambda}$.
Hence, the following theorem easily follows.

\begin{theorem}
There exists triangular grids $G_6$, $G_6'$ and $G_6''$ whose longest dipath is 5 such that $\vec{\lambda}(G_6)=4$, $\vec{\lambda}(G_6')=5$ and $\vec{\lambda}(G_6'')=6$.
\end{theorem}

In order to study other values of $\vec{\lambda}$ in function of the length of the maximum dipath, we prove the following result.

\begin{lemma}
\label{lemma.dist2}
If in any triangular grid $G_6$ there exists a node $v$ whose {\em all} its adjacent nodes are at mutual distance $\leq 2$, then $\vec{\lambda}(G_6) \geq 7$. If there exist three such nodes, and they are at distance 2 each other, then $\vec{\lambda}(G_6) \geq 8$.
\end{lemma}
\begin{proof} %controllata
The best case is that $v$ has color either 0 or $\vec{\lambda}$. So it forbids 2 colors to the six $v$'s adjacent nodes, that must receive different colors, as they are at mutual distance $\leq 2$. It follows that at least 8 colors are necessary, i.e. $\vec{\lambda} \geq 7$.

Now, assume that there are three nodes $v$, $v'$ and $v''$ whose adjacent nodes are at mutual distance 2, such that the oriented distances between $v$ and $v'$, $v$ and $v''$, $v'$ and $v''$ are 2.
Node $v$ and its adjacent nodes can be labeled using colors from 0 to 7, in view of the first part of this lemma and, without loss of generality, the label assigned to $v$ is 0.
Even $v'$ and its adjacent nodes can be labeled using colors from 0 to 7 but, in this case, we have to arrange color in order to assign color 7 to $v'$, as it cannot be labeled with the same color as $v$.
We cannot continue this reasoning for $v''$ and its adjacent nodes, as $v''$ cannot be labeled with 0 or 7, that are the colors of $v$ and $v'$, respectively, so it must receive a color forbidding three labels for its adjacent nodes. It follows that color 8 must be used, i.e. $\vec{\lambda} \geq 8$. \qed
\end{proof}

We want to exploit this lemma in order to build a triangular grid requiring at least 8 colors and having longest dipath of minimal length.

By proceeding in a exhaustive way, we can construct all possible orientations of the non-oriented grid in Figure \ref{fig.triangular}.a, that is the minimal graph having a central node and 6 adjacent nodes, imposing that all nodes are at mutual distance $\leq 2$. 
In this way, it turns out that every such orientation has its longest dipath of length $\geq 5$. One of these grids is depicted in Figure \ref{fig.triangular5}.a.
We conclude this argument observing that the nodes lying on the opposite sides of the hexagon are joint by a length 2 dipath passing through the central node and this cannot be avoided. 

In order to apply the second part of Lemma \ref{lemma.dist2}, we consider three triangular grids constituted by a central node and its 6 adjacent nodes; the central nodes $v$, $v'$ and $v''$ must be connected by dipaths of length 2. 
If such dipaths are all disjoint (see Figure \ref{fig.triangular5}.b), then at least two of them must be oriented in the same direction and form a $P_5$. 
This $P_5$ must be in fact a subpath of a $P_7$, due to the presence of the length 2 dipaths passing through the central node.
If, on the contrary, the dipaths are not disjoint (see Figure \ref{fig.triangular5}.c), then it is easy to see that at most one edge can be shared, and again they form a $P_5$ that is a subpath of a $P_7$.
So, in both cases, the built graph has its longest dipath of length $\geq 6$.

\begin{figure}[ht]
\begin{center}
    \epsfxsize=11.5cm
	\epsfbox{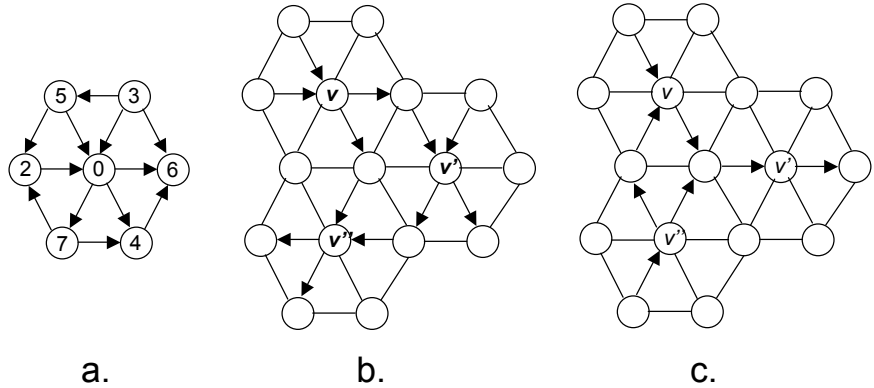}
   \caption{a. A triangular grid with the length of the longest dipath equal to five having $\vec{\lambda}=7$. b. and c. Constructions for a triangular grid with the length of the longest dipath greater than or equal to six having $\vec{\lambda}\geq 8$.}
   \protect\label{fig.triangular5}
\end{center}
\end{figure}

These reasoning leads us to claim the following

\begin{theorem}
There exist triangular grids $G_6$ and $G_6'$, whose longest dipaths are 5 and $\geq 6$, respectively such that $\vec{\lambda}(G_6)=7$ and $\vec{\lambda}(G_6') = 8$.
\end{theorem}

Supported by the previous reasoning and many attempts, the author is convinced that it does not exist any triangular grid $G_6$ whose longest dipath has length 4 such that $\vec{\lambda}(G_6)\geq 7$ nor any triangular grid $G_6'$ whose longest dipath has length 5 such that $\vec{\lambda}(G_6')\geq 8$.
In other words, the conjecture is that any triangular grids $G_6$ and $G_6'$ such that $\vec{\lambda}(G_6)\geq 7$ and $\vec{\lambda}(G_6)= 8$ must have their longest dipath of length at least 5 and $\geq 6$, respectively.

%%%%%%%%%%%%%
\section{Hexagonal Grids}
%%%%%%%%%%%%%

Any hexagonal grid is a bipartite graph, so if the longest dipath has length 3 then $3 \leq \vec{\lambda} \leq 4$.

\begin{theorem}
There exist hexagonal grids $G_3$ and $G_3'$, whose longest dipath is 3 such that $\vec{\lambda}(G_3)=3$ and $\vec{\lambda}(G_3')=4$.
There exists an hexagonal grid $G_3''$ whose longest dipath is 4 such that $\vec{\lambda}(G_3'')=4$.
\end{theorem}

\begin{proof}
The grids shown in Figure \ref{fig.hexagonal34}.a and  \ref{fig.hexagonal34}.b have their longest dipath of length 3, while the grid in  Figure \ref{fig.hexagonal34}.c has its longest dipath of length 4.
The depicted labelings are feasible. The optimality can be proved by an exhaustive reasoning and hence omitted for the sake of brevity. \qed
\end{proof}

\begin{figure}[ht]
\begin{center}
    \epsfxsize=11.5cm
	\epsfbox{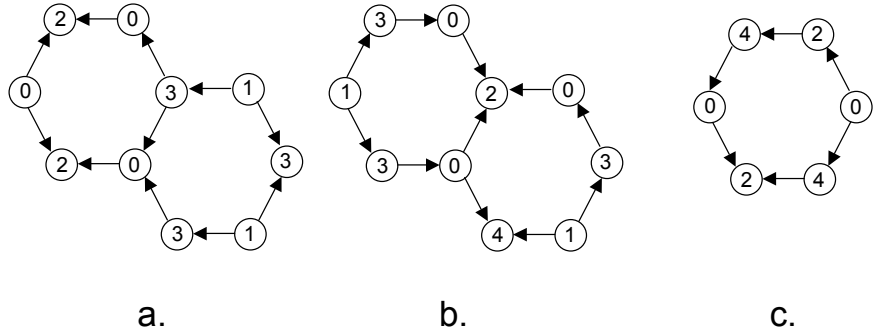}
   \caption{a. and b. Two hexagonal grids with the length of the longest dipath equal to three having $\vec{\lambda}=3$ and $\vec{\lambda}=4$, respectively; c.aA hexagonal grid with the length of the longest dipath equal to four having $\vec{\lambda}=4$. Optimal labelings are depicted.}
   \protect\label{fig.hexagonal34}
\end{center}
\end{figure}

When the length of the longest dipath is $\geq 4$, then we cannot exploit previous results on bipartite graphs. It is known that $\vec{\lambda} \leq \lambda =5$ but the author has not been able to find any orientation of a hexagonal grid such that $\vec{\lambda}=5$. It is convincement of the author that the reason is that no orientation can oblige all the three neighbors of a node to have three different colors.

Due to this consideration, we strongly suspect that no orientation of any hexagonal grid requires $\vec{\lambda} =5$.

%%%%%%%%%%%%%%%%%%%%%
\section{Conclusions and Open Problems}

In this paper we have approached the oriented $L(2,1)$-labeling problem on regular grids, i.e. squared, triangular and hexagonal grids, showing the results divided by length of the longest dipath. 
More precisely, we have evaluated the smallest value of the length of the longest dipath such that a certain value of $\vec{\lambda}$ holds.

Along the paper the author has claimed some convincements, whose the third one is the most important, as if it was true, it would distinguish the behavior of the hexagonal tiling with respect to the squared and triangular tilings. Indeed, in the case of squared and triangular grids, we provided witnesses for every possible value of $\vec{\lambda}$ and we have just conjectured that in order to have certain values of $\vec{\lambda}$, the longest dipath must have a certain minimum length. 
On the contrary, in the case of hexagonal grids, the author thinks that  there is a value  of $\vec{\lambda}$ that is not attainable at all.

We think that the reason of the difference between hexagonal tiling from the one side and squared and triangular tilings from the other side lies in the concept of {\em girth}, i.e. the length of the shortest cycle, of the underlying graph. 
Indeed, informally speaking, if the girth $g$ is small (i.e. $g \leq 4$), together with the regularity of the graph, it is somehow possible to guarantee that all the neighbors of a node get different colors, for opportune orientations, and hence the value of $\vec{\lambda}$ can be increased up to optimal span in the undirected case $\lambda$. This is not possible anymore when $g \geq 5$ as there is no way to guarantee that all the neighbors of a certain node are at mutual distance $\leq 2$. 
So, the author thinks that -- more than depending on the longest dipath length -- the value of $\vec{\lambda}$ could depend on $g$. 
We highlight that the previous papers on this topic could not deal with this parameter, as they restrict the problem to trees (no girth) or to bipartite graphs whose longest dipath has length at most three ($g \leq 4$).

We conclude this paper with a further and very general conjecture:

\begin{conjecture}
Every planar directed graph whose underlying graph has girth $g \geq 5$ has $\vec{\lambda} \leq 5$.
\end{conjecture}

The author hopes that this latter conjecture will open new research lines in order to explore the oriented $L(2,1)$-labeling problem more in deep.

%%%%%%%%%%%%%%%%%%%%%%%%%%%%%%%%%
%bibliografia


\begin{thebibliography}{99}

\begin{small}
%%%si%%%%%%%%%%%%%%%%%%%%%
\bibitem{Aal01}
K.I. Aardal, S.P.M. van Hoesel, A.M.C.A. Koster, C. Mannino and A. Sassano.
Models and Solution Techniques for Frequency Assignment Problems.
{\em ZIB-Report 01-40, Konrad-Zuse-Zentrum fur
Informationstechnik Berlin}, 2001.

%%%si%%%%%%%
\bibitem{C06}
T. Calamoneri. The $L(h,k)$-Labelling Problem: a Survey and Annotated Bibliography.
{\em The Computer Journal}, 49(5): 585--608, 2006.
A continuously updated version is freely available at {\tt http://www.dsi.uniroma1.it/~calamo/survey.html}

%%%%si
\bibitem{CP04}%journal 6 cit.
T. Calamoneri  and  R. Petreschi.
$L(h,1)$-Labeling Subclasses of Planar Graphs.
{\em Journal on Parallel and Distributed Computing},  64(3): 414--426, 2004. 

%%%%si
\bibitem{CL03}
G. J. Chang and S.-C. Liaw. 
The $L(2,1)$-labeling problem on ditrees.
{\em Ars Combinatoria}, 66: 23--31, 2003.

%%%%si
\bibitem{Cal07}
G.J. Chang, J.-J. Chen, D. Kuo and S.-C. Liaw.
Distance-two labelings of digraphs.
{\em Discrete Applied Mathematics}, 155: 1007--1013, 2007.

%%%%si%%%%%%%%%
\bibitem{GY92}
J.R. Griggs and R.K. Yeh.
Labeling graphs with a Condition at Distance 2.
{\em SIAM J. Disc. Math}, 5:586--595, 1992.

%%%%si%%%%%%%%%%
\bibitem{H80}%journal 1 cit.
W.K. Hale.
Frequency assignment: theory and applications.
 {\em Proceedings of IEEE},  68: 1497--1514, 1980.

\bibitem{Y90}%Ph.D. thesis=book 1 cit.
R.K. Yeh. 
{\em Labeling Graphs with a Condition at Distance Two.}
Ph.D. Thesis, University of South Carolina, Columbia, South Carolina, 1990. 

%%%si%%%%%%%%%%%%%
\bibitem{Y06}%2 cit.
R.K. Yeh. 
A Survey on Labeling Graphs with a Condition at Distance Two.
{\em Discrete Mathematics}, 306: 1217--1231, 2006.

\end{small}
\end{thebibliography}
\end{document}